\documentclass[11pt]{article}

\newcommand{\prl}{Phys. Rev. Lett.~}
\newcommand{\pra}{Phys. Rev. A~}

\newcommand{\pla}{Phys. Lett. A~}

\usepackage{appendix}
\usepackage{epstopdf}
\usepackage{color}
\usepackage{subfigure}
\usepackage[T1]{fontenc}
\usepackage[sc]{mathpazo}
\usepackage{amsmath}
\usepackage{amssymb}
\usepackage{graphicx,color}
\usepackage{framed}
\usepackage{multirow}
\usepackage{enumerate}
\usepackage{amsthm}
\usepackage{amsfonts,mathrsfs}
\usepackage{geometry} 
\usepackage{eepic}
\usepackage{ifthen}
\usepackage[vcentermath]{youngtab}
\usepackage[unicode=true,pdfusetitle, bookmarks=true,bookmarksnumbered=false,bookmarksopen=false, breaklinks=false,pdfborder={0 0 0},backref=false,colorlinks=false] {hyperref}
\hypersetup{
colorlinks,linkcolor=myurlcolor,citecolor=myurlcolor,urlcolor=myurlcolor}
\definecolor{myurlcolor}{rgb}{0,0,0.7}

\usepackage{color}

\newcommand{\blue}{\textcolor{blue}}

\newcommand{\proj}[1]{| #1\rangle\!\langle #1 |}

\usepackage{hyperref}
\hypersetup{pdfpagemode=UseNone}

\newcommand{\tinyspace}{\mspace{1mu}}

\newcommand{\op}[1]{\operatorname{#1}}

\newcommand{\abs}[1]{\left\lvert\tinyspace #1 \tinyspace\right\rvert}

\renewcommand{\t}{{\scriptscriptstyle\mathsf{T}}}

\newcommand{\setft}[1]{\mathrm{#1}}
\newcommand{\lin}[1]{\setft{L}\left(#1\right)}
\newcommand{\density}[1]{\setft{D}\left(#1\right)}

\renewcommand{\vec}{\op{vec}}

\newcommand{\rank}{\op{rank}}

\def \diag {\mathrm{diag}}

\def\complex{\mathbb{C}}
\def\real{\mathbb{R}}

\def\I{\mathbb{1}}

\newenvironment{mylist}[1]{\begin{list}{}{
    \setlength{\leftmargin}{#1}
    \setlength{\rightmargin}{0mm}
    \setlength{\labelsep}{2mm}
    \setlength{\labelwidth}{8mm}
    \setlength{\itemsep}{0mm}}}
    {\end{list}}

\def\ot{\otimes}

\newcommand{\out}[2]{| #1\rangle\langle #2 |}
\newcommand{\Inner}[2]{\left\langle #1 , #2\right\rangle}
\newcommand{\Innerm}[3]{\left\langle #1 \left| #2 \right| #3 \right\rangle}

\newcommand{\Pa}[1]{\left(#1\right)}

\newcommand{\Br}[1]{\left[#1\right]}
\newcommand{\set}[1]{\{#1\}}
\newcommand{\Set}[1]{\left\{#1\right\}}

\newcommand{\ket}[1]{|#1\rangle}

\DeclareMathOperator{\trace}{Tr}

\newcommand{\Ptr}[2]{\trace_{#1}\Pa{#2}}

\newcommand{\Tr}[1]{\Ptr{}{#1}}

\newcommand{\Abs}[1]{\left|\tinyspace#1\tinyspace\right|}

\def\cH{\mathcal{H}}
\def\cN{\mathcal{N}}

\def\cX{\mathcal{X}}\def\cY{\mathcal{Y}}

\def\bsA{\boldsymbol{A}}\def\bsD{\boldsymbol{D}}
\def\bsF{\boldsymbol{F}}

\def\bsU{\boldsymbol{U}}\def\bsV{\boldsymbol{V}}\def\bsW{\boldsymbol{W}}\def\bsX{\boldsymbol{X}}

\def\bsp{\boldsymbol{p}}
\def\bsu{\boldsymbol{u}}\def\bsx{\boldsymbol{x}}

\def\U{\textsf{U}}

\def\mC{\mathbb{C}}

\newtheorem{thrm}{Theorem}[section]

\newtheorem{prop}[thrm]{Proposition}
\newtheorem{cor}[thrm]{Corollary}
\theoremstyle{definition}

\numberwithin{equation}{section}

\newcounter{questionnumber}

\begin{document}

\title{A characterization of maximally entangled two-qubit states}

\author{\blue{Junjun Duan}$^1$\,, \blue{Lin Zhang}$^1$\footnote{E-mail: godyalin@163.com}\,, \blue{Quan Qian}$^1$\,, \blue{Shao-Ming Fei}$^{2,3}$\\
  {\it\small $^1$School of Sciences, Hangzhou Dianzi University, Hangzhou 310018, PR~China}
  \\{\it\small $^2$Max Planck Institute for Mathematics in the Sciences, 04103 Leipzig, Germany}\\
  {\it\small $^3$School of Mathematical Sciences, Capital Normal University, Beijing, 100048, PR~China}}
\date{}
\maketitle

\begin{abstract}
As already known by Rana's result
\href{https://doi.org/10.1103/PhysRevA.87.054301}{[\pra {\bf87}
(2013) 054301]}, all eigenvalues of any partial-transposed bipartite
state fall within the closed interval $[-\frac12,1]$. In this note,
we study a family of bipartite quantum states whose minimal
eigenvalues of partial-transposed states being $-\frac12$. For a
two-qubit system, we find that the minimal eigenvalue of its
partial-transposed state is $-\frac12$ if and only if such two-qubit
state must be maximally entangled. However this result does not hold
in general for a two-qudit system when the dimensions of the
underlying space are larger than two.
\end{abstract}


\section{Introduction}

Let $\rho_{AB}$ be a quantum state in a bipartite quantum system
$\cH_A\ot\cH_B$, then the positive partial transpose (PPT) criterion
indicates that, for any separable state $\rho_{AB}$, it must hold
$\rho^{\t_A}_{AB}\geqslant0$, where $^{\t_A}$ denotes the partial
transpose on subsystem $A$. This PPT condition is firstly proposed
by Peres \cite{Peres1996}. Such condition is not only a necessary
but also a sufficient one for separability in qubit-qubit sytem and
qubit-qutrit/qutrit-qubit system cases \cite{Horodecki1996}. The PPT
condition can also be verified from the moments of the randomized
measurements \cite{Gray2018,Elben2020,Zhou2020}.

Recently, Yu \emph{et al} \cite{Yu2021} found that the PPT condition
can be studied by considering the so-called \emph{partial transpose
moments} (PT-moments)
\begin{eqnarray*}
p_k := \Tr{\Br{\rho^{\t_A}_{AB}}^k}.
\end{eqnarray*}
In fact, these quantities can be efficiently measured in experiments
\cite{Elben2020,Zhou2020}. To see the basic idea behind the
PT-moments-based entanglement detection, suppose that we know all
the PT-moments $\bsp^{(d)}=(p_1,\ldots,p_d)$, where $d=d_Ad_B$ is
the dimension of the global system $\cH_A\ot\cH_B$, where
$d_{A/B}=\dim(\cH_{A/B})$. We call $\bsp^{(d)}$ the PT-moment vector
of the state $\rho_{AB}$. Denote by $(x_1,\ldots,x_d)$ all the
eigenvalues of $\rho^{\t_A}_{AB}$. As already known,
$(x_1,\ldots,x_d)$ determines completely the elementary symmetric
polynomials $(e_1,\ldots,e_d)$, where $e_1=\sum^d_{k=1}x_k$,
$e_2=\sum_{1\leqslant i<j\leqslant d}x_ix_j,\ldots$, and
$e_d=\prod^d_{k=1}x_k$; and reversely $(e_1,\ldots,e_d)$ can
determine $(x_1,\ldots,x_d)$ when ignoring their order. In fact,
$(e_1,\ldots,e_d)$ and $(p_1,\ldots,p_d)$, where $p_k$'s are the
power sum of $x_i$'s, necessarily identify each other via the
following relationship between $e_k$ and $p_k$ \cite{Mac1995}:
\begin{eqnarray*}
p_k = \Abs{\begin{array}{ccccc}
             e_1 & 1 & 0 & \cdots & 0 \\
             2e_2 & e_1 & 1 & \cdots & 0 \\
             \vdots & \vdots & \vdots & \ddots & \vdots \\
             (k-1)e_{k-1} & e_{k-2} & e_{k-3} & \cdots & 1\\
             ke_k & e_{k-1} & e_{k-2} & \cdots & e_1
           \end{array}
}\quad (k\geqslant1)
\end{eqnarray*}
and
\begin{eqnarray*}
e_k = \frac1{k!}\Abs{\begin{array}{ccccc}
             p_1 & 1 & 0 & \cdots & 0 \\
             p_2 & p_1 & 2 & \cdots & 0 \\
             \vdots & \vdots & \vdots & \ddots & \vdots \\
             p_{k-1} & p_{k-2} & p_{k-3} & \cdots & k-1\\
             p_k & p_{k-1} & p_{k-2} & \cdots & p_1
           \end{array}
}\quad (k\geqslant1).
\end{eqnarray*}
Therefore $(p_1,\ldots,p_d)$ determines $(x_1,\ldots,x_d)$ up to
their order. Then, all eigenvalues of the partial-transposed state
$\rho^{\t_A}_{AB}$ can be directly obtained. Based on the above
observation, the PPT criterion can be verified immediately. For
convenience, we always assume that $p_1=1$. In addition, $p_2$ is
just the purity due to the fact that
$\Tr{\Br{\rho^{\t_A}_{AB}}^2}=\Tr{\rho^2_{AB}}$. In \cite{Yu2021},
the authors studied the following problem.

\noindent{\bf PT-Moment problem:} Given the PT-moments of order $n$,
is there a separable state compatible with the data? More
technically formulated, \emph{given} the PT-moment vector
\begin{eqnarray*}
\bsp^{(n)}=(p_1,\ldots,p_n),
\end{eqnarray*}
\emph{is there} a separable quantum state $\rho_{AB}$ such that
\begin{eqnarray*}
p_k = \Tr{\Br{\rho^{\t_A}_{AB}}^k},\quad k=1,\ldots,n?
\end{eqnarray*}
It is natural to consider the detection of entanglement in
$\rho_{AB}$ from a few of the PT-moments due to the difficulty in
measuring all the PT-moments, such as \cite{Yu2021}. Note that the
partial-transposed state $\rho^{\t_A}_{AB}$, for
$\rho_{AB}\in\density{\complex^m\ot\complex^n}$, the set of all
bipartite quantum states acting on $\complex^m\ot\complex^n$, cannot
have more than $(m-1)(n-1)$ number of negative eigenvalues and all
eigenvalues of $\rho^{\t_A}_{AB}$ fall within $[-\tfrac12,1]$
\cite{Rana2013pra}. Using the second PT-moment to bound the third
one is an interesting question. Moreover, we find that this method
can be used to get a characterization of maximally entangled
two-qubit states, that is, e.g., a two-qubit state is maximally
entangled if and only if the minimal eigenvalue of its
partial-transposed state is $-\frac12$, this is also equivalently to
the condition that $\bsp^{(4)}=(1,1,\tfrac14,\tfrac14)$, the
PT-moment vector of the two-qubit state $\rho_{AB}$. This amounts to
give a criterion of maximally entangled states using PT-moment
vector whose components are measurable quantities.


\section{Main result}

In this section, our question can be essentially asked: Is
$\rho_{AB}$ maximally entangled if
$\rho_{AB}\in\density{\mC^2\ot\mC^2}$ and the minimal eigenvalue of
its partial-transposed state $\rho^{\t_A}_{AB}$ is
$\lambda_{\min}\Pa{\rho^{\t_A}_{AB}}=-\frac12$? We give the positive
answer to this question in our main result , i.e.,
Theorem~\ref{th:A1}. To that end, we obtain the proof by showing a
series of propositions.

\begin{prop}
Let $\rho_{AB},\sigma_{AB},\tau_{AB}\in\density{\mC^m\ot\mC^n}$,
where $\rho_{AB}=t\sigma_{AB}+(1-t)\tau_{AB}$ for some $t\in(0,1)$.
If $\lambda_{\min}\Pa{\rho^{\t_A}_{AB}}=-\frac12$, i.e., the minimal
eigenvalue of $\rho^{\t_A}_{AB}$, then
$\lambda_{\min}\Pa{\sigma^{\t_A}_{AB}}=\lambda_{\min}\Pa{\tau^{\t_A}_{AB}}=-\frac12$.
\end{prop}

\begin{proof}
Using the main result in \cite{Rana2013pra}, we see that, for any
bipartite state $\varrho_{AB}\in\density{\mC^m\ot\mC^n}$, we have
that
\begin{eqnarray*}
\lambda_{\min}(\varrho^{\t_A}_{AB})\in\Br{-\tfrac12,1}.
\end{eqnarray*}
Thus
\begin{eqnarray*}
\lambda_{\min}(\sigma^{\t_A}_{AB})\geqslant-\frac12,\quad
\lambda_{\min}(\tau^{\t_A}_{AB})\geqslant-\frac12.
\end{eqnarray*}
As already known, there exists a pure state
$\ket{\psi_0}\in\mC^m\ot\mC^n$, corresponding to the minimal
eigenvalue $\lambda_{\min}\Pa{\rho^{\t_A}_{AB}}$, such that
\begin{eqnarray*}
-\frac12&=&\lambda_{\min}\Pa{\rho^{\t_A}_{AB}} =
\Innerm{\psi_0}{\rho^{\t_A}_{AB}}{\psi_0} = t
\Innerm{\psi_0}{\sigma^{\t_A}_{AB}}{\psi_0}+
(1-t)\Innerm{\psi_0}{\tau^{\t_A}_{AB}}{\psi_0}\\
&\geqslant&t
\lambda_{\min}\Pa{\sigma^{\t_A}_{AB}}+(1-t)\lambda_{\min}\Pa{\tau^{\t_A}_{AB}}\geqslant
-\frac12.
\end{eqnarray*}
We must have that
$\lambda_{\min}\Pa{\sigma^{\t_A}_{AB}}=\Innerm{\psi_0}{\sigma^{\t_A}_{AB}}{\psi_0}=\lambda_{\min}\Pa{\tau^{\t_A}_{AB}}=\Innerm{\psi_0}{\tau^{\t_A}_{AB}}{\psi_0}=-\frac12$.
\end{proof}

\begin{cor}
Suppose $\rho_{AB}\in\density{\mC^m\ot\mC^n}$ has the pure state
decomposition: $\rho_{AB}=\sum_k\lambda_k\proj{\psi_k}$, where
$\lambda_k>0$ for all indices $k$. If
$\lambda_{\min}\Pa{\rho^{\t_A}_{AB}} =-\frac12$, then
\begin{eqnarray*}
\lambda_{\min}\Pa{\psi^{\t_A}_k} = - \frac12.
\end{eqnarray*}
Here $\psi_k:=\proj{\psi_k}$.
\end{cor}

Recall that there is a correspondence between the set
$\lin{\cY,\cX}$ of all linear operators from a finite-dimensional
Hilbert space $\cY$ to another finite-dimensional Hilbert space
$\cX$. It can be explained immediately. Denote by $\cX\ot\cY$ the
tensor space of $\cX$ and $\cY$. Let the orthonormal bases of $\cX$
and $\cY$ be $\Set{\ket{i}:i=1,\ldots,\dim(\cX)}$ and
$\Set{\ket{j}:j=1,\ldots,\dim(\cY)}$, respectively. The mentioned
correspondence between $\lin{\cY,\cX}$ and $\cX\ot\cY$ is defined by
the linear mapping $\vec:\lin{\cY,\cX}\to \cX\ot\cY$ via
$\vec(\out{i}{j})=\ket{ij}$ for all $i,j$ \cite{Watrous2018}.

Let $\ket{\psi}\in\mC^d\ot\mC^d$ be a bipartite pure state. Then
there is an $d\times d$ complex matrix $\bsX$ such that
$\ket{\psi}=\vec(\bsX)$. By Singular Value Decomposition
($\mathrm{SVD}$), there are two unitary matrices
$\bsU,\bsV\in\mathsf{U}(d)$ such that $\bsX=\bsU\Sigma\bsV^\dagger$,
where $\Sigma=\diag(\sigma_1,\ldots,\sigma_r,\ldots,\sigma_d)$ for
$\sigma_1\geqslant\cdots\geqslant\sigma_d\geqslant0$ and
$r=\rank(\bsX)\leqslant d$. Note that $\sum^d_{j=1}\sigma^2_j=1$.
Then
\begin{eqnarray*}
\proj{\psi} =
\bsU\ot\overline{\bsV}\vec(\Sigma)\vec(\Sigma)^\dagger(\bsU\ot\overline{\bsV})^\dagger
\end{eqnarray*}
implying that
\begin{eqnarray*}
\proj{\psi}^{\t_A} =
\overline{\bsU}\ot\overline{\bsV}\Pa{\vec(\Sigma)\vec(\Sigma)^\dagger}^{\t_A}(\overline{\bsU}\ot\overline{\bsV})^\dagger.
\end{eqnarray*}
Due to the fact that $\Sigma=\sum^d_{i=1}\sigma_i\proj{i}$, we see
that
\begin{eqnarray*}
\vec(\Sigma)\vec(\Sigma)^\dagger =
\sum^d_{i,j=1}\sigma_i\sigma_j\out{ij}{ij},\quad
\Pa{\vec(\Sigma)\vec(\Sigma)^\dagger}^{\t_A} =
\sum^d_{i,j=1}\sigma_i\sigma_j\out{ji}{ij}
\end{eqnarray*}
\begin{prop}
All eigenvalues of $\proj{\psi}^{\t_A}$ is given by
$\Set{\sigma^2_1,\ldots,\sigma^2_d; \pm\sigma_i\sigma_j(1\leqslant
i< j\leqslant d)}$.
\end{prop}

\begin{proof}
Let $\bsF =
\Pa{\vec(\Sigma)\vec(\Sigma)^\dagger}^{\t_A}=\sum^d_{i,j=1}\sigma_i\sigma_j\out{ji}{ij}$.
Then
\begin{eqnarray*}
\bsF\bsF^\dagger
&=&\Pa{\sum^d_{i,j=1}\sigma_i\sigma_j\out{ji}{ij}}\Pa{\sum^d_{k,l=1}\sigma_k\sigma_l\out{lk}{kl}}^\dagger=
\sum^d_{i,j,k,l=1}\sigma_i\sigma_j\sigma_k\sigma_l\out{ji}{ij}\cdot\out{kl}{lk}\\
&=&\sum^d_{i,j,k,l=1}\delta_{ik}\delta_{jl}\sigma_i\sigma_j\sigma_k\sigma_l\out{ji}{lk}=\sum^d_{i,j=1}(\sigma_i\sigma_j)^2\out{ji}{ji},
\end{eqnarray*}
that is, $\abs{\bsF}=\sqrt{\bsF\bsF^\dagger}
=\sum^d_{i,j=1}\sigma_i\sigma_j\out{ji}{ji}$. Note that
\begin{eqnarray*}
\Tr{\bsF} = \sum^d_{i=1}\sigma^2_i,\quad \Tr{\abs{\bsF}} =
\sum^d_{i,j=1}\sigma_i\sigma_j.
\end{eqnarray*}
By using the Jordan decomposition $\bsF=\bsF_+-\bsF_-$, where
$\bsF^\dagger_\pm=\bsF_\pm\geqslant0$ and
$\bsF_+\bsF_-=0=\bsF_-\bsF_+$. Then
\begin{eqnarray*}
\Tr{\bsF_+}-\Tr{\bsF_-} = \sum^d_{i=1}\sigma^2_i,\quad
\Tr{\bsF_+}+\Tr{\bsF_-} = \sum^d_{i,j=1}\sigma_i\sigma_j,
\end{eqnarray*}
and
\begin{eqnarray*}
\Tr{\bsF_+} =
\sum^d_{i=1}\sigma^2_i+\sum_{i<j}\sigma_i\sigma_j,\quad \Tr{\bsF_-}
= \sum_{i<j}\sigma_i\sigma_j.
\end{eqnarray*}
Therefore all eigenvalues of $\bsF$ is given by
$\Set{\sigma^2_1,\ldots,\sigma^2_d; \pm\sigma_i\sigma_j(1\leqslant
i< j\leqslant d)}$.
\end{proof}

Thus, we need to characterize those bipartite pure states with the
minimal eigenvalue of its partial transposed state being
$-\tfrac12$.

\begin{prop}
If $\ket{\psi}\in\mC^2\ot\mC^2$ is a pure state, then
$\lambda_{\min}\Pa{\psi^{\t_A}} = - \frac12$, where
$\psi\equiv\proj{\psi}$, if and only if $\ket{\psi}$ is a maximally
entangled state, i.e., $\ket{\psi}$ is proportional to the locally
unitarily rotation of the vector $\vec(\mathbb{I}_2)$. Here
$\mathbb{I}$ is the identity operator.
\end{prop}

\begin{proof}
Let $\ket{\psi}\in\mC^2\ot\mC^2$. Suppose
$\bsx=(x_1,x_2,x_3,x_4)\in\real^4$ is the eigenvalues of
$\psi^{\t_A}$ with $1\geqslant x_1\geqslant x_2\geqslant
x_3\geqslant x_4\geqslant-\frac12$ \cite{Rana2013pra}. Clearly
$x_4=\lambda_{\min}\Pa{\psi^{\t_A}}$.

Now let $x_4=-\frac12$. Again, we see that $x_3\geqslant0$ by Rana's
result. Denote $p_k=\Tr{\Br{\psi^{\t_A}}^k}$, where $k=1,2,\ldots$.
It is easily see that $p_1=1,p_2=1$. Then\
\begin{eqnarray*}
\begin{cases}
1=x_1+x_2+x_3+\Pa{-\frac12}\\
1=x^2_1+x^2_2+x^2_3+\Pa{-\frac12}^2.
\end{cases}
\end{eqnarray*}
Due to the constraint $1\geqslant x_1\geqslant x_2\geqslant
x_3\geqslant0$, the above system of equations has the unique
solution: $x_1=x_2=x_3=\frac12$.

What we now have proved is that if
$\lambda_{\min}\Pa{\psi^{\t_A}}=-\frac12$ for some pure state
$\ket{\psi}\in\mC^2\ot\mC^2$, then all eigenvalues of $\psi^{\t_A}$
is $\set{\frac12,\frac12,\frac12,-\frac12}$. In fact, if
$\lambda_{\min}\Pa{\rho^{\t_A}_{AB}}=-\frac12$ for some state
$\rho_{AB}\in\density{\mC^2\ot\mC^2}$, then all eigenvalues of
$\rho^{\t_A}_{AB}$ is $\set{\frac12,\frac12,\frac12,-\frac12}$.

For such pure state $\ket{\psi}\in \mC^2\ot\mC^2$, there exists a
$2\times 2$ complex matrix $\bsA$ such that $\ket{\psi}=\vec(\bsA)$.
By SVD of $\bsA$, we get that $\bsA=\bsU\bsD\bsV^\dagger$ where
$\bsD=\diag(s_0,s_1)$ with $s_0\geqslant s_1\geqslant0$ and
$\bsU,\bsV\in\U(2)$. Then
\begin{eqnarray*}
\proj{\psi} &=&
\bsU\ot\overline{\bsV}\vec(\bsD)\vec(\bsD)^\dagger\Pa{\bsU\ot\overline{\bsV}}^\dagger\\
&=&\bsU\ot\overline{\bsV}\vec(\bsD)\vec(\bsD)^\dagger
\bsU^\dagger\ot\bsV^\t,
\end{eqnarray*}
where
\begin{eqnarray*}
\vec(\bsD)\vec(\bsD)^\dagger =
\vec\Pa{\sum_is_i\proj{i}}\vec\Pa{\sum_js_j\proj{j}}^\dagger =
\sum^1_{i,j=0}s_is_j\out{ii}{jj}.
\end{eqnarray*}
Next we establish the equations concerning $(s_0,s_1)$. The first
one is $s^2_0+s^2_1=1$ due to the fact that
$\Tr{\bsD^2}=\Inner{\psi}{\psi}=1$. The second one is
\begin{eqnarray*}
\proj{\psi}^{\t_A} &=& \sum^1_{i,j=0}s_is_j
(\bsU\out{i}{j}\bsU^\dagger)^{\t}\ot\overline{\bsV}\out{i}{j}\bsV^\t\\
&=& \sum^1_{i,j=0}s_is_j
\overline{\bsU}\out{j}{i}\bsU^\t\ot\overline{\bsV}\out{i}{j}\bsV^\t\\
&=& \Pa{\overline{\bsU}\ot\overline{\bsV}}\sum^1_{i,j=0}s_is_j
\out{j}{i}\ot\out{i}{j}\Pa{\overline{\bsU}\ot\overline{\bsV}}^\dagger\\
&=&\Pa{\overline{\bsU}\ot\overline{\bsV}}\sum^1_{i,j=0}s_is_j
\out{ji}{ij}\Pa{\overline{\bsU}\ot\overline{\bsV}}^\dagger.
\end{eqnarray*}
Now both $\proj{\psi}^{\t_A}$ and $\sum^1_{i,j=0}s_is_j
\out{ji}{ij}$ has the same eigenvalues. That is, all eigenvalues of
\begin{eqnarray*}
\sum^1_{i,j=0}s_is_j \out{ji}{ij} = \Pa{
\begin{array}{cccc}
 s_1^2 & 0 & 0 & 0 \\
 0 & 0 & s_1 s_2 & 0 \\
 0 & s_1 s_2 & 0 & 0 \\
 0 & 0 & 0 & s_2^2 \\
\end{array}
}
\end{eqnarray*}
is
$\set{s^2_0,s^2_1,s_0s_1,-s_0s_1}=\set{\frac12,\frac12,\frac12,-\frac12}$.
This implies that
\begin{eqnarray*}
s^2_0=s^2_1=s_0s_1=\frac12.
\end{eqnarray*}
Its unique solution is given by
$(s_0,s_1)=\Pa{\frac1{\sqrt{2}},\frac1{\sqrt{2}}}$. Therefore
$\bsA=\frac1{\sqrt{2}}\bsU\bsV^\dagger$. Then
\begin{eqnarray*}
\ket{\psi} =\vec(\bsA) = \frac1{\sqrt{2}}\vec(\bsU\bsV^\dagger)=
\frac1{\sqrt{2}}\bsU\ot\overline{\bsV}\vec(\mathbb{\I}_2).
\end{eqnarray*}
We have proved that $\ket{\psi}$ is a maximally entangled state.
Conversely, if $\ket{\psi}$ is a maximally entangled state, then the
minimal eigenvalue of its partial-transposed state is apparently
$-\frac12$ \cite{Li2012,Zhao2015}.
\end{proof}
For the partial-transposed maximally entangled states in
$\complex^n\ot\complex^n$, the eigenvalues must be $\pm\frac1n$
where the multiplicities of $\frac1n$ and $-\frac1n$ are
$\frac{n(n+1)}2$ and $\frac{n(n-1)}{2}$, respectively. Thus its
PT-moment vector is given by
\begin{eqnarray*}
\bsp^{(n^2)} = (p_1,\ldots,p_{n^2}),\quad p_k =
\frac{(n+1)+(n-1)(-1)^k}{2n^{k-1}}.
\end{eqnarray*}
In particular, for the case where $n=2$, the PT-moment vector
$\bsp^{(4)} = (1,1,\tfrac14,\tfrac14)$.

\begin{thrm}\label{th:A1}
Let $\rho_{AB}\in\density{\mC^2\ot\mC^2}$ be a quantum state, then
the following statements are equivalent:
\begin{enumerate}[(i)]
\item the PT-moment vector of $\rho_{AB}$ is $\bsp^{(4)}=
(1,1,\tfrac14,\tfrac14)$.
\item $\rho_{AB}$ must be maximally entangled.
\end{enumerate}
\end{thrm}

\begin{proof}
For the implication (ii) $\Longrightarrow$ (i), the proof is
trivially. Next, we show that (i) implies (ii). Given (i). Let
$(x_1,x_2,x_3,x_4)$, where $x_1\geqslant x_2\geqslant x_3\geqslant
x_4$, be eigenvalues of the partial-transposed state
$\rho^{\t_A}_{AB}$, then Rana's result \cite{Rana2013pra} means that
$1\geqslant x_1\geqslant x_2\geqslant
x_3\geqslant\begin{cases}0\\x_4\geqslant-\frac12\end{cases}$. By the
given PT-moment vector, we see that
\begin{eqnarray*}
\begin{cases}
p_1 = x_1+x_2+x_3+x_4=1\\
p_2 = x^2_1+x^2_2+x^2_3+x^2_4=1\\
p_3 = x^3_1+x^3_2+x^3_3+x^3_4=\frac14\\
p_4 = x^4_1+x^4_2+x^4_3+x^4_4=\frac14
\end{cases}
\end{eqnarray*}
In fact, note that
\begin{eqnarray*}
e_k = \frac1{k!}\Abs{\begin{array}{ccccc}
             p_1 & 1 & 0 & \cdots & 0 \\
             p_2 & p_1 & 2 & \cdots & 0 \\
             \vdots & \vdots & \vdots & \ddots & \vdots \\
             p_{k-1} & p_{k-2} & p_{k-3} & \cdots & k-1\\
             p_k & p_{k-1} & p_{k-2} & \cdots & p_1
           \end{array}
}\quad (k\geqslant1),
\end{eqnarray*}
we see that $e_1=1,e_2=0,e_3=-\frac14,e_4=-\frac1{16}$. The
characteristic polynomial of $\rho^{\t_A}_{AB}$ is given by
\begin{eqnarray*}
f(x) &=& x^4-e_1x^3+e_2x^2-e_3x+e_4 = x^4-x^3+\frac14x-\frac1{16}\\
&=& \frac1{16} (2x-1)^3 (2x+1).
\end{eqnarray*}
Solving this system of equations via $f(x)=0$, we get that
$x_1=x_2=x_3=\frac12$ and
$x_4=\lambda_{\min}(\rho^{\t_A}_{AB})=-\frac12$. Next denote by
$\bsF=\sum^1_{i,j=0}\out{ji}{ij}$, then $\bsF\ket{ij}=\ket{ji}$ and
$\bsF^{\t_A}=\vec(\mathbb{I})\vec(\mathbb{I})^\dagger$. Note that
$\bsF\frac{\ket{01}-\ket{10}}{\sqrt{2}}=-\frac{\ket{01}-\ket{10}}{\sqrt{2}}$.
Take $\ket{\bsx}=\frac{\ket{01}-\ket{10}}{\sqrt{2}}$, and we get
that $\Innerm{\bsx}{\bsF}{\bsx}=-1=\lambda_{\min}(\bsF)$.

From the previous discussion, we see that in the pure state
decomposition of $\rho_{AB}$:
$\rho_{AB}=\sum^{N-1}_{k=0}\lambda_k\proj{\psi_k}$, where
$\lambda_0\geqslant\lambda_1\geqslant\cdots\geqslant\lambda_{N-1}\geqslant0$,
all pure state $\ket{\psi_k}$ must be maximally entangled state.
Then there exist a pure state $\ket{\bsu}$ such that
\begin{eqnarray*}
\Innerm{\bsu}{\psi^{\t_A}_k}{\bsu}=-\frac12\quad (k=0,1,\ldots,N-1).
\end{eqnarray*}
There exist $\bsU_k,\bsV_k\in\U(2)$ such that
\begin{eqnarray*}
\ket{\psi_k}=\frac1{\sqrt{2}}\vec(\bsU_k\bsV^\dagger_k)
=\frac1{\sqrt{2}}\bsU_k\bsV^\dagger_k\ot\mathbb{I}\vec(\mathbb{I})
=\frac1{\sqrt{2}}\bsW_k\ot\mathbb{I}\vec(\mathbb{I}),
\end{eqnarray*}
where $\bsW_k=\bsU_k\bsV^\dagger_k$, implying that $\psi^{\t_A}_k =
\frac12\Pa{\overline{\bsW}_k\ot\mathbb{I}}\bsF
\Pa{\overline{\bsW}_k\ot\mathbb{I}}^\dagger$. Now let
$\ket{\bsu_k}=\bsW^\t_k\ot\mathbb{I}\ket{\bsu}$. Then
\begin{eqnarray*}
-\frac12 =
\Innerm{\bsu}{\psi^{\t_A}_k}{\bsu}=\frac12\Innerm{\bsu_k}{\bsF}{\bsu_k}\quad
(k=0,1,\ldots,N-1).
\end{eqnarray*}
That is,
\begin{eqnarray*}
\lambda_{\min}(\bsF)=-1=\Innerm{\bsu_k}{\bsF}{\bsu_k}\quad
(k=0,1,\ldots,N-1).
\end{eqnarray*}
Because $-1$ is the simple eigenvalue of $\bsF$, the eigenspace
corresponding to $-1$ is just $\mC\ket{\bsx}$. This indicates that
all $\ket{\bsu_k}=e^{\mathrm{i}\theta_k}\ket{\bsx}$ due to the
normalization of $\ket{\bsu_k}$. Furthermore
\begin{eqnarray*}
\ket{\bsu} =
e^{\mathrm{i}\theta_k}\Pa{\overline{\bsW}_k\ot\mathbb{I}}\ket{\bsx}\quad
(k=0,1,\ldots,N-1).
\end{eqnarray*}
In fact, the phase factor $e^{\mathrm{i}\theta_k}$ can be absorbed
into the unitary matrix $\bsW_k$. Without loss of generality, we
assume that
\begin{eqnarray*}
\ket{\bsu} =\Pa{\overline{\bsW}_k\ot\mathbb{I}}\ket{\bsx}\quad
(k=0,1,\ldots,N-1).
\end{eqnarray*}
Because there is a matrix $\bsX$ such that $\ket{\bsx}=\vec(\bsX)$,
then
\begin{eqnarray*}
\bsX = \frac1{\sqrt{2}}(\out{0}{1}-\out{1}{0}) =
\frac1{\sqrt{2}}\Pa{\begin{array}{cc}
                      0 & 1 \\
                      -1 & 0
                    \end{array}
}.
\end{eqnarray*}
It is easily seen that $\bsX$ is invertible. We see that
\begin{eqnarray*}
\ket{\bsu} =\vec\Pa{\overline{\bsW}_k\bsX}\quad (k=0,1,\ldots,N-1).
\end{eqnarray*}
implying that $\overline{\bsW}_0\bsX = \overline{\bsW}_1\bsX
=\cdots= \overline{\bsW}_{N-1}\bsX$, i.e., due to the fact that
$\bsX$ is invertible, then $\bsW_0 = \bsW_1=\cdots= \bsW_{N-1}$, or
\begin{eqnarray*}
\bsU_0\bsV^\dagger_0 = \bsU_1\bsV^\dagger_1=\cdots=
\bsU_{N-1}\bsV^\dagger_{N-1},
\end{eqnarray*}
implying that
$\vec(\bsU_0\bsV^\dagger_0)=\vec(\bsU_1\bsV^\dagger_1)=\cdots=\vec(\bsU_{N-1}\bsV^\dagger_{N-1})$,
that is, $\ket{\psi_0}=\ket{\psi_1}=\cdots=\ket{\psi_{N-1}}$.
Therefore $\rho_{AB} = \sum_k\lambda_k\proj{\psi_k}=\proj{\psi_0}$
is a maximally entangled state.
\end{proof}

In fact, our main result Theorem~\ref{th:A1} tells us that the
PT-moment vector of a two-qubit state $\rho_{AB}$ is
$(1,1,\tfrac14,\tfrac14)$ iff the minimal eigenvalue of its
partial-transposed state $\rho^{\t_A}_{AB}$ is $-\tfrac12$ iff
$\rho_{AB}$ is maximally entangled. Naturally, we would expect a
similar relation between the magnitude of the lowest negative
eigenvalue of the partial-transposed state and the maximally
entangled states in higher-dimensional underlying spaces. However,
the following result, i.e., Proposition~\ref{prop:maxent}, indicates
that the minimal eigenvalue of the partial-transposed maximally
entangled state would approach zero when the dimension of the
underlying space becomes larger and larger. Indeed, after tedious
computations and induction, we can draw the following conclusion:
\begin{prop}\label{prop:maxent}
Let $\rho_{AB}\in\density{\mC^n\ot\mC^n}$ be a quantum state. If the
PT-moment vector of $\rho_{AB}$ is $\bsp^{(n^2)} =
(p_1,\ldots,p_{n^2})$, where $p_k =
\frac{(n+1)+(n-1)(-1)^k}{2n^{k-1}}$. Then
$\lambda_{\min}\Pa{\rho^{\t_A}_{AB}}=-\frac1n$.
\end{prop}

\begin{proof}
As an illustration, for a two-qutrit system
$\complex^3\ot\complex^3$ as an example, we get that
$(e_1,\ldots,e_9)$ from $\bsp^{(9)}=(p_1,\ldots,p_9)$, where
$p_1=p_2=1,p_3=p_4=\frac19,p_5=p_6=\frac1{81},p_7=p_8=\frac1{729}$
and $p_9=\frac1{6561}$. That is,
\begin{eqnarray*}
e_1=1,e_2=0,e_3=-\frac8{27},e_4=-\frac2{27},e_5=\frac2{81},\\
e_6=\frac8{729},e_7=0,e_8=-\frac1{2187},e_9=-\frac1{19683}.
\end{eqnarray*}
Furthermore, the characteristic polynomial of $\rho^{\t_A}_{AB}$ is
given by
\begin{eqnarray*}
f(x) &=& x^9-x^8+\frac{8 x^6}{27}-\frac{2 x^5}{27}-\frac{2
x^4}{81}+\frac{8 x^3}{729}-\frac{x}{2187}+\frac{1}{19683}\\
&=& \frac{(3 x-1)^6 (3 x+1)^3}{19683}.
\end{eqnarray*}
Thus we get that
$$
x_1=\cdots=x_6=\frac13,x_7=x_8=x_9=-\frac13.
$$
Therefore $\lambda_{\min}(\rho^{\t_A}_{AB})=-\frac13$.
\end{proof}
From the above result, when $n\to\infty$,
$\lambda_{\min}(\rho^{\t_A}_{AB})\to0$ for maximally entangled state
in $\complex^n\ot\complex^n$. Based on the observation, we can
conclude that the family of bipartite states with the minimal
eigenvalue of their partial-transposed states being $-\frac12$ is
different from the set of maximally entangled states when the
dimensions of the underlying spaces are larger than two. This also
indicates that the magnitude of the only lowest negative eigenvalue
of the partial-transposed state in higher-dimensional space would
not be enough to identify the maximally entangled state when there
are more than one negative eigenvalues. In fact, it is known that
for higher dimensions the characterization of entanglement can be
given by the so-called "negativity" \cite{Plenio2005}, which is
defined as the absolute values of the sum of all the negative
eigenvalues of the partial-transposed state. That is, the negativity
of $\rho_{AB}$ is given by $\cN(\rho_{AB}) =
\Abs{\sum_i\min\Set{\lambda_i(\rho^{\t_A}_{AB}),0}}$. With this
notion, our Theorem~\ref{th:A1} can be rewritten as: For two-qubit
state $\rho_{AB}$, $\cN(\rho_{AB})=\frac12$ iff $\rho_{AB}$ is
maximally entangled. The success of such characterization lies at
the possible number of negative eigenvalues being at most one. The
reason for failure of this result in high-dimensional space is that
only one negative eigenvalue (the lowest one) would not be enough to
characterize entanglement when there could be more than one negative
eigenvalue.

\section{Concluding remarks}

In this short note, we make an attempt to study the structure of a
family of bipartite states with the extremal eigenvalue being
$-\frac12$ of its partial-transposed states. We employ the approach
in studying PT-moments by Yu \emph{et al} recently, i.e., PT-moment
vectors, to get a characterization of maximally entangled two-qubit
states. In higher dimensional system, we are curious about the
problem that the PT-moment vector $\bsp^{(n^2)} =
(p_1,\ldots,p_{n^2})$, where $p_k =
\frac{(n+1)+(n-1)(-1)^k}{2n^{k-1}}$, generated by the maximally
entangled states, whether only corresponds the maximally entangled
states. Clearly, a bipartite state in
$\density{\complex^n\ot\complex^n}$ with
$\lambda_{\min}(\rho^{\t_A}_{AB})=-\frac12$ is, in general, not
maximally entangled state unless $n=2$. In the future research, we
will continue to figure out the structure of the mentioned family of
states, especially, find out the connection between such family of
states and the maximally entangled states in higher dimension.
Furthermore, we will study the connection between the entanglement
existing in bipartite states and the number of negative eigenvalues
of the corresponding partial-transposed states.\\~\\
\noindent\textbf{Author Contributions:} Writing--original draft,
J.D., L.Z. and Q.Q.; Writing--review \& editing, S.-M.F. All authors have read and agreed to the published version of the manuscript.\\~\\
\noindent\textbf{Funding:} This work is supported by the National
Natural Science Foundation of China under Grant Nos. (11971140,
12075159, 12171044); Beijing Natural Science Foundation (Grant No.
Z190005); Academy for Multidisciplinary Studies, Capital Normal
University; Shenzhen Institute for Quantum Science and Engineering,
Southern University of Science and Technology (No.
SIQSE202001), the Academician Innovation Platform of Hainan Province.\\~\\
\noindent\textbf{Conflicts of Interest:} The authors declare no
conflict of interest.


\end{document}